\documentclass{acmtrans2m}

\usepackage{front,front_acmtrans2m_vd}

\usepackage{amsmath,amssymb,atmath}                      

\usepackage{graphicx}

\usepackage{url,subfig}

\begin{document}

\newcommand{\mytitle}{%
Fast minimum-weight double-tree shortcutting for Metric TSP:
Is the best one good enough?}

\newcommand{\myabstract}{%
The Metric Traveling Salesman Problem (TSP) 
is a classical NP-hard optimization problem.
The double-tree shortcutting method for Metric TSP 
yields an exponentially-sized space of TSP tours,
each of which approximates the optimal solution 
within at most a factor of $2$.
We consider the problem of finding among these tours
the one that gives the closest approximation,
i.e.\ the \emph{minimum-weight double-tree shortcutting}. 
Burkard et al.\ gave an algorithm for this problem,
running in time $O(n^3+2^d n^2)$ and memory $O(2^d n^2)$,
where $d$ is the maximum node degree in the rooted minimum spanning tree.
We give an improved algorithm for the case of small $d$
(including planar Euclidean TSP, where $d \leq 4$),
running in time $O(4^d n^2)$ and memory $O(4^d n)$.
This improvement allows one to solve the problem 
on much larger instances than previously attempted.
Our computational experiments suggest
that in terms of the time-quality tradeoff,
the minimum-weight double-tree shortcutting method
provides one of the best known tour-constructing heuristics.}




\myfront

\markboth{%
V.~Deineko and A.~Tiskin}{%
Fast minimum-weight double-tree shortcutting for Metric TSP}

\section{Introduction}
\label{s-intro}

The Metric Travelling Salesman Problem (TSP) 
is a classical combinatorial optimization problem.
We represent a set of $n$ points in a metric space
by a complete weighted graph on $n$ nodes,
where the weight of an edge is defined 
by the distance between the corresponding points.
The objective of Metric TSP is to find in this graph 
a minimum-weight Hamiltonian cycle
(equivalently, a minimum-weight tour visiting every node at least once).
The most common example of Metric TSP is the planar Euclidean TSP, 
where the points lie in the two-dimensional Euclidean plane,
and the distances are measured according to the Euclidean metric.

Metric TSP, even restricted to planar Euclidean TSP, 
is well-known to be NP-hard \cite{Papadimitriou:77}.
Metric TSP is also known to be NP-hard to approximate 
to within a ratio $1.00456$,
but polynomial-time approximable to within a ratio $1.5$.
Fixed-dimension Euclidean TSP is known to have a PTAS
(i.e.\ a family of algorithms 
with approximation ratio arbitrarily close to $1$) \cite{Arora:98};
this generalises to any metric 
defined by a fixed-dimension Minkowski vector norm.

Two simple approaches, 
the \emph{double-tree method} \cite{Rosenkrantz+:77}
and the \emph{Christofides method} \cite{Christofides:76,Serdyukov:78},
allow one to approximate the solution of Metric TSP 
within a factor of $2$ and $1.5$, respectively.
Both methods belong to the class of \emph{tour-constructing heuristics},
i.e.\ ``heuristics that incrementally construct a tour 
and stop as soon as a valid tour is created'' \cite{Johnson_McGeoch:02}.
In both methods, we build an Eulerian graph on the given point set,
select an Euler tour of the graph,
and then perform \emph{shortcutting} on this tour
by removing repeated nodes, until all node repetitions are removed.
In general, it is not prescribed which one of several occurrences
of a particular node to remove.
Therefore, the methods yield an exponentially-sized space of TSP tours
(shortcuttings of a specific Euler tour in a specific Eulerian graph),
each approximating the optimal solution 
within a factor of $2$ (respectively, $1.5$).

The two methods differ in the way 
the initial weighted Eulerian graph is constructed.
Both start by finding the graph's minimum-weight spanning tree (MST).
The double-tree method then doubles every edge in the MST,
while the Christofides method adds to the MST a minimum-weight matching 
built on the set of odd-degree nodes\@.
The weight of the resulting Euler tour 
exceeds the weight of the optimal TSP tour
by at most a factor of $2$ (respectively, $1.5$),
and the subsequent shortcutting can only decrease the tour weight.

While any tour obtained by shortcutting of the original Euler tour
approximates the optimal solution within the specified factor,
clearly, it is still desirable to find 
the shortcutting that gives the closest approximation.
Given an Eulerian graph on a set of points,
we will consider its \emph{minimum-weight shortcutting}
across all shortcuttings of all possible Euler tours of the graph.
We shall correspondingly speak 
about \emph{the minimum-weight double-tree}
and \emph{the minimum-weight Christofides} methods.

Unfortunately, for general Metric TSP, both the double-tree and Christofides
minimum-weight shortcutting problems are NP-hard.
Consider an instance of the Hamiltonian cycle problem on an unweighted graph;
this can be regarded as an instance of Metric TSP with weights 1 and 2.
Add an extra node connected to all the original nodes by edges of weight 1,
and take the newly added edges as the MST\@.
It is easy to see that 
the resulting minimum-weight double-tree shortcutting problem
is equivalent to the original Hamiltonian cycle problem.
The minimum-weight double-tree shortcutting problem
was believed for a long time to be NP-hard even for planar Euclidean TSP,
until a polynomial-time algorithm 
was given by \citeN{Burkard+:98}.
This is the algorithm we improve upon in the current paper.
In contrast, the minimum-weight Christofides shortcutting problem 
remains NP-hard even for planar Euclidean TSP
\cite{Papadimitriou_Vazirani:84}.

In the rest of this paper, we will mainly deal with the \emph{rooted MST},
which is obtained from the MST 
by selecting an arbitrary node as \emph{the root}.
In the rooted MST, the terms \emph{parent}, \emph{child}, 
\emph{ancestor}, \emph{descendant}, \emph{sibling}, \emph{leaf} 
all have their standard meaning.
Let $d$ denote the maximum number of children per node in the rooted MST.
Note that in the Euclidean plane, 
the maximum degree of an unrooted MST is at most $6$.
Moreover, a node can have degree equal to $6$,
only if it is surrounded by six equidistant nodes
forming a regular hexagon;
we can exclude this degenerate case from consideration
by a slight perturbation of the input points.
This leaves us with an unrooted MST of maximum degree $5$.
By choosing a node of degree less than $5$ as the root,
we obtain a rooted MST with $d \leq 4$.

The minimum-weight double-tree shortcutting algorithm
of \cite{Burkard+:98} applies to the general Metric TSP,
and runs in time $O(n^3+2^d n^2)$ and memory $O(2^d n^2)$.
In this paper, we give an improved algorithm%
\footnote{Note that Burkard et al.\ \cite{Burkard+:98} also
give an $O(2^d n^3)$ algorithm for a more general TSP-type problem,
where the set of admissible tours is restricted by a given PQ-tree.
Our algorithm does not improve on the algorithm of \cite{Burkard+:98}
for this more general problem.}
for the case of small $d$,
running in time $O(4^d n^2)$ and memory $O(4^d n)$.
In the planar Euclidean case, both above algorithms
run in polynomial time and memory.

We then describe our implementation of the new algorithm,
which incorporates a couple of additional heuristic improvements
designed to speed up the algorithm and to increase its approximation quality.
Computational experiments show that the approximation quality 
and running time of our implementation
are among the best known tour-constructing heuristics.

A preliminary version of this paper appeared as \cite{Deineko_Tiskin:07}.

\section{The algorithm}
\label{s-alg}

\subsection{Preliminaries}

Let $G$ be a weighted graph
representing the Metric TSP problem on $n$ points.
The double-tree method consists of the following stages:
\begin{itemize}
\item construct the minimum spanning tree of $G$;
\item duplicate every edge of the tree, obtaining an $n$-node Eulerian graph;
\item select an Euler tour of the double-tree graph;
\item reduce the Euler tour to a Hamiltonian cycle 
by repeated \emph{shortcutting}, 
i.e. replacing a node sequence $a,b,c$ by $a,c$,
as long as node $b$ appears elsewhere in the current tour.
\end{itemize}
We say that a Hamiltonian cycle \emph{conforms} 
to the doubled spanning tree,
if it can be obtained from that tree by shortcutting one of its Euler tours.
We also extend this definition to paths,
saying that a path conforms to the tree,
if it is a subpath of a conforming Hamiltonian cycle.

In our minimum-weight double-tree shortcutting algorithm,
we refine the bottom-up dynamic programming approach of \cite{Burkard+:98}.
Initially, we select an arbitrary node $r$ as the root of the tree.
For a node $u$, we denote by $C(u)$ the set of all children of $u$,
and by $T(u)$ the node set of the maximal subtree rooted at $u$,
i.e.\ the set of all descendants of $u$ (including $u$ itself).
For a set of siblings $U$,
we denote by $T(U)$ the (disjoint) union of all subtrees $T(u)$, $u \in U$.
When $U$ is empty, $T(U)$ is also empty.

The characteristic property of a conforming Hamiltonian cycle is as follows:
for every node $u$, 
the cycle must contain all nodes of $T(u)$ consecutively in some order.
For an arbitrary node set $S$, 
we will say that a path through the graph \emph{sweeps} $S$,
if it visits all nodes of $S$ consecutively in some order.
In this terminology, a conforming Hamiltonian cycle
must, for every node $u$, contain a subpath sweeping the subtree $T(u)$.

In the rest of this section, we denote
the metric distance between $u$ and $v$ by $d(u,v)$.
We use the symbol $\uplus$ to denote disjoint set union.
For brevity, given a set $A$ and an element $a$,
we write $A \uplus a$ instead of $A \uplus \{a\}$,
and $A \setminus a$ instead of $A \setminus \{a\}$.

\subsection{Upsweep: Computing solution weight}

The algorithm proceeds by computing minimum-weight sweeping paths 
in progressively increasing subtrees,
beginning with the leaves and finishing with the whole tree $T(r)$.
A similar approach is adopted in \cite{Burkard+:98},
where in each subtree, all-pairs minimum-weight sweeping paths are computed.
In contrast, our algorithm only computes 
single-source minimum-weight sweeping paths originating at the subtree's root.
This leads to substantial savings in time and memory.

A non-root node $v \in C(u)$ is \emph{active},
if its subtree $T(v)$ has already been processed,
but its parent's subtree $T(u)$ has not yet been processed.
In every stage of the algorithm, we choose \emph{the current node} $u$,
so that all children of $u$ (if any) are active.
We call $T(u)$ \emph{the current subtree}.
Let $V \subseteq C(u)$, $a \in T(V)$.
We denote by $D^u_V(a)$ the weight of the shortest conforming path
starting from $u$, sweeping the subtree $u \uplus T(V)$, and finishing at $a$.

Consider the current subtree $T(u)$.
Processing this subtree will yield 
the values $D^u_V(a)$ for all $V \subseteq C(u)$, $a \in T(V)$.
In order to process the subtree, we need the corresponding values
for all subtrees rooted at the children of $u$.
More precisely, we need the values $D^v_W(a)$
for every child $v \in C(u)$, every subset $W \subseteq C(v)$, 
and every destination node $a \in T(W)$.
We do not need any explicit information on subtrees 
rooted at grandchildren and lower descendants of $u$.

Given the current subtree $T(u)$, the values $D^u_V(a)$
are computed inductively for all sets $V$ of children of $u$.
The induction is on the size of the set $V$.
The base of the induction is trivial:
no values $D^u_V(a)$ exist when $V = \emptyset$.

In the inductive step, given a set $V \subseteq C(u)$,
we compute the values $D^u_{V \uplus v}(a)$
for all $v \in C(u) \setminus V$, $a \in T(v)$, as follows.
By the inductive hypothesis, 
we have the values $D^u_V(a)$ for all $a \in T(V)$.
The main part of the inductive step consists 
in computing a set of auxiliary values $D^u_{V,W}(v)$,
for all subsets $W \subseteq C(v)$.
Every such value represents the weight 
of the shortest conforming path 
starting from node $u$, sweeping the subtree $u \uplus T(V)$, 
then sweeping the subtree $T(W) \uplus v$, and finishing at node $v$.
Suppose the path exits the subtree $u \uplus T(V)$ at node $x$
and enters the subtree $T(W) \uplus v$ at node $y$.
We have
\begin{figure}[tbp]
\centering
\includegraphics[bb=148 581 291 661]{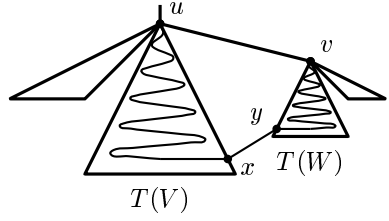}
\caption{\label{f-v} Computation of $D^u_{V,W}(v)$}
\end{figure}
\begin{equation}
D^u_{V,W}(v) =
\begin{cases}
d(u,v) & \text{if $V = \emptyset$, $W = \emptyset$} \\
\min_{y \in T(W)} \bigbra{d(u,y) + D^v_W(y)} &
  \text{if $V = \emptyset$, $W \neq \emptyset$} \\
\min_{x \in T(V)} \bigbra{D^u_V(x) + d(x,v)} &
  \text{if $V \neq \emptyset$, $W = \emptyset$} \\
\min_{x \in T(V); y \in T(W)} \bigbra{D^u_V(x) + d(x,y) + D^v_W(y)} &
  \text{if $V \neq \emptyset$, $W \neq \emptyset$}
\end{cases}
\label{eq-v}
\end{equation}
(see \figref{f-v}).
The required values $D^v_W(y)$ have been obtained previously,
while processing subtrees $T(v)$ for the active nodes $v \in C(u)$.
Note that the computed auxiliary values 
include $D^u_{V \uplus v}(v) = D^u_{V,C(v)}(v)$.

Now we can compute the values $D^u_{V \uplus v}(a)$ 
for all $a \in T(v) \setminus v = T(C(v))$.
A path corresponding to $D^u_{V \uplus v}(a)$ 
must sweep $u \uplus T(V)$, and then $T(v)$, finishing at $a$.
While in $T(v)$, the path will first sweep 
a (possibly single-node) subtree $v \uplus T(W)$, finishing at $v$.
Then, starting at $v$, the path will sweep 
the subtree $v \uplus T(\overline W)$, 
where $\overline W = C(v) \setminus W$, finishing at $a$.
Considering every possible disjoint bipartitioning
$W \uplus {\overline W} = C(v)$, such that $a \in T(\overline W)$,
we have
\begin{figure}[tbp]
\centering

\subfloat[Case $W=\emptyset$]{\label{f-a1}%
\includegraphics[bb=148 581 284 661]{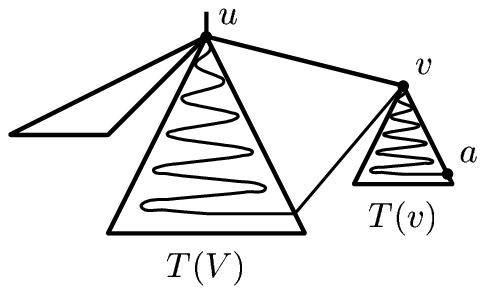}}
\qquad
\subfloat[Case $W\neq \emptyset$]{\label{f-a2}%
\includegraphics[bb=148 581 284 661]{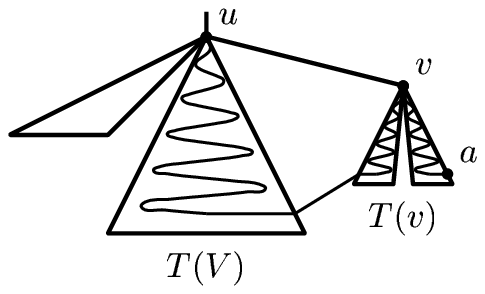}}

\caption{\label{f-a}%
Computation of $D^u_{V \uplus v}(a)$, $a \in T(v)$}
\end{figure}
\begin{equation}
D^u_{V \uplus v}(a) = 
\min_{W \uplus \overline W = C(v) :\: a \in T(\overline W)}
\bigbra{ D^u_{V,W}(v) + D^v_{\overline W}(a) }
\label{eq-a}
\end{equation}
(see \figref{f-a}).

We now have the values $D^u_{V \uplus v}(a)$ for all $a \in T(v)$.
The computation \eqref{eq-v}--\eqref{eq-a} 
is repeated for every node $v \in C(u) \setminus V$.
The inductive step is now completed.

The processing of subtree $T(u)$ terminates 
when all possible choices of subset $V$ and node $v$ have been exhausted.

Eventually, the root $r$ of the tree becomes the current node,
and we process the complete tree $T(r)$.
This establishes the values $D^r_S(a)$ 
for all $S \subseteq C(r)$, $a \in T(S)$,
which includes the values $D^r_{C(r)}(a)$ for all $a \neq r$.
The weight of the minimum-weight conforming Hamiltonian cycle
can now be determined as
\begin{equation}
\min_{a \neq r} \bigbra{ D^r_{C(r)}(a) + d(a,r) }
\label{eq-root}
\end{equation}

\begin{theorem}
\label{th-upsweep}
The upsweep algorithm computes
the weight of the minimum-weight tree shortcutting 
in time $O(4^d n^2)$ and space $O(2^d n)$.
\end{theorem}
\begin{proof}

In computation \eqref{eq-v},
the total number of quadruples $u,v,x,y$ is at most $n^2$
(since for every pair $x$, $y$,
the node $u$ is determined uniquely as the lowest common ancestor of $x$, $y$,
and the node $v$ is determined uniquely 
as a child of $u$ and an ancestor of $y$).
In computation \eqref{eq-a},
the total number of triples $u,v,a$ is also at most $n^2$
(since for every pair $u$, $a$,
the node $v$ is determined uniquely 
as a child of $u$ and an ancestor of $y$).
For every such quadruple or triple,
the computation is performed at most $4^d$ times,
corresponding to $2^d$ possible choices of each of $V$, $W$.
The cost of computation \eqref{eq-root} is negligible.
Therefore, the total time complexity of the algorithm is $O(4^d n^2)$.

Since our goal at this stage is just to compute the solution weight,
at any given moment we only need to store the values $D^u_V(a)$,
where $u$ is either an active node, or the current node 
(i.e.\ the node for which these values are currently being computed).
When $u$ corresponds to an active node,
the number of possible pairs $u,a$ is at most $n$
(since node $u$ is determined uniquely 
as the root of the active subtree containing $a$).
When $u$ corresponds to the current node,
the number of possible pairs $u,a$ is also at most $n$
(since node $u$ is fixed).
For every such pair, we need to keep at most $2^d$ values,
corresponding to $2^d$ possible choices of $V$.
The remaining space costs are negligible.
Therefore, the total space complexity of the algorithm is $O(2^d n)$.
\qed
\end{proof}

\subsection{Downsweep: Reconstructing full solution}

In order to reconstruct the minimum-weight Hamiltonian cycle itself, 
we must keep all the auxiliary values $D^u_{V,W}(v)$
obtained in the course of the upsweep computation
for every parent-child pair $u, v$.
We solve recursively the following problem: 
given a node $u$, a set $V \subseteq C(u)$, and a node $a \in T(V)$,
find the minimum-weight path $P^u_V(a)$ starting from $u$, 
sweeping subtree $u \uplus T(V)$, and finishing at $a$.
To compute the global minimum-weight Hamiltonian cycle,
it is sufficient to determine the path $P^r_{C(r)}(a)$,
where $r$ is the root of the tree, and $a$ is the node 
for which the minimum in \eqref{eq-root} is attained.

\begin{figure}[tbp]
\centering
\includegraphics[bb=148 567 362 661]{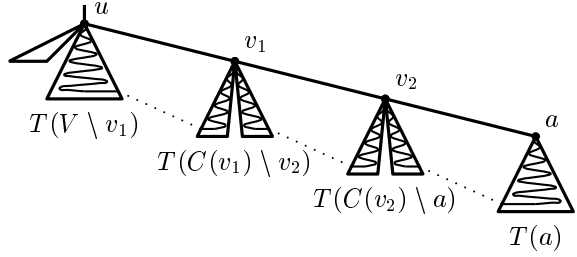}
\caption{\label{f-back}
Computation of $P^u_V(a)$, $a \in T(V)$, $k=3$}
\end{figure}

For any $u$, $V \subseteq C(u)$, $a \in T(V)$, 
consider the (not necessarily conforming or minimum-weight)
path $u=v_0 \to v_1 \to v_2 \to\dots\to v_k=a$,
joining nodes $u$ and $a$ in the tree (see \figref{f-back}).
The conforming minimum-weight path $P^u_V(a)$ 
first sweeps the subtree $u \uplus T(V \setminus v_1)$.
After that, for every node $v_i$, $0 < i < k$, the path $P^u_V(a)$ 
sweeps the subtree $v_i \uplus T(C(v_i) \setminus v_{i+1})$
as follows:
first, it sweeps a subtree $v_i \uplus T(W_i)$, 
finishing at $v_i$,
and then, starting at $v_i$, 
it sweeps the subtree $v_i \uplus T({\overline W}_i)$,
for some disjoint bipartitioning 
$W_i \uplus {\overline W_i} = C(v_i) \setminus v_{i+1}$.
Finally, the path $P^u_V(a)$ sweeps the subtree $T(a)$,
finishing at $a$.

The optimal choice of bipartitionings can be found as follows.
We construct a weighted directed \emph{layered graph} 
with a source vertex corresponding to node $u=v_0$,
a sink vertex corresponding to node $v_k=a$,
and $k-1$ intermediate layers of vertices, 
each layer corresponding to a node $v_i$, $0 < i < k$.
Each intermediate layer consists of at most $2^{d-1}$ vertices,
representing all different disjoint bipartitionings 
of the node set $C(v_i) \setminus v_{i+1}$.
The source and the sink vertices represent the trivial bipartitionings
$\emptyset \uplus (V \setminus v_1) = V \setminus v_1$ and
$C(a) \uplus \emptyset = C(a)$, respectively.
Every consecutive pair of vertex layers 
(including the source and the sink vertices)
are fully connected by forward arcs.
In particular, the arc from a vertex representing the bipartitioning
$X \uplus {\overline X}$ in layer $i$,
to the vertex representing the bipartitioning
$Y \uplus {\overline Y}$ in layer $i+1$,
is given the weight $D^{v_i}_{{\overline X}, Y}(v_{i+1})$.
It is easy to see that an optimal choice of bipartitioning corresponds to 
the minimum-weight path from the source to the sink in the layered graph.
This minimum-weight path can be found 
by a standard dynamic programming algorithm
(such as the Bellman--Ford algorithm, see e.g.\ \cite{Cormen+:01})
in time proportional to the number of arcs in the layered graph.

Let $W_1 \uplus {\overline W}_1, \dots, W_{k-1} \uplus {\overline W}_{k-1}$
now denote the $k-1$ obtained optimal subtree bipartitionings.
The $k$ arcs of the corresponding source-to-sink shortest path 
determine $k$ edges (not necessarily consecutive)
in the minimum-weight sweeping path $P^u_V(a)$.
These edges are shown in \figref{f-back} by dotted lines.
It now remains to apply the downsweep algorithm 
recursively in each of the subtrees 
$u \uplus T(V \setminus v_1)$, 
$v_1 \uplus T(W_1)$, $v_1 \uplus T({\overline W}_1)$, 
$v_2 \uplus T(W_2)$, $v_2 \uplus T({\overline W}_2)$, \ldots,
$v_{k-1} \uplus T(W_{k-1})$, $v_{k-1} \uplus T({\overline W}_{k-1})$, 
$T(a)$.

\begin{theorem}
\label{th-downsweep}
Given the output and the necessary intermediate values 
of the upsweep algorithm, the downsweep algorithm computes the edges 
of the minimum-weight tree shortcutting in time and space $O(4^d n)$.
\end{theorem}
\begin{proof}
The construction of the layered graph and the minimum-weight path computation
runs in time $O(4^d k)$,
where $k$ is the number of edges in the tree path 
$u=v_0 \to v_1 \to v_2 \to\dots\to v_k=a$
in the current level of recursion.
Since the tree paths in different recursion levels are edge-disjoint,
the total number of edges in these paths is at most $n$.
Therefore, the time complexity of the downsweep algorithm is $O(4^d n)$.

By \thref{th-upsweep}, 
the space complexity of the upsweep algorithm is $O(2^d n)$.
In addition to the storage used internally by the upsweep algorithm,
we also need to keep all the values $D^u_{V,W}(v)$.
The number of possible pairs $u,v$ is at most $n$
(since node $u$ is determined uniquely as the parent of $v$).
For every such pair, we need to keep at most $4^d$ values,
corresponding to $2^d$ possible choices of each of $V$, $W$.
The remaining space costs are negligible.
Therefore, the total space complexity of the downsweep algorithm is $O(4^d n)$.
\qed
\end{proof}

\section{Heuristic improvements}
\label{s-heur}

Despite the guaranteed approximation ratio
of the double-tree shortcutting and Christofides methods,
neither has performed well in previous computational experiments
(see \cite{Johnson_McGeoch:97,Reinelt:94}).
However, to our knowledge, none of these experiments
explored the minimum-weight double-tree shortcutting approach.
Instead, the double-tree shortcutting was performed
in some suboptimal, easily computable order,
such as a depth-first tree traversal.
We shall call this method \emph{depth-first double-tree shortcutting}.

In particular, \cite{Reinelt:94} compares 37 tour-constructing heuristics, 
including the depth-first double-tree algorithm
and the Christofides algorithm,
on a set of 24 geometric instances from the TSPLIB database \cite{Reinelt:91}.
Although most instances in this experiment are quite small
(2000 or fewer points),
they still allow us to make some qualitative judgement 
about the approximation quality of different heuristics.
Depth-first double-tree shortcutting
turns out to have the lowest quality of all 37 heuristics,
while the quality of the Christofides algorithm is somewhat higher, 
but still far from the top.

Intuitively, it is clear that the reason 
for the poor approximation quality of the two algorithms 
may be in the wrong choice of the shortcutting order,
especially considering 
that the overall number of alternative choices is typically exponential.
This observation motivated us to implement
the minimum-weight double-tree shortcutting algorithm from \cite{Burkard+:98}.
It came as no surprise that this algorithm 
showed higher approximation quality
than all the tour constructing heuristics in Reinelt's experiment.
Unfortunately, Reinelt's experiment did not account 
for the running time of the algorithms under investigation.
The theoretical time complexity 
of the previous minimum-weight double-tree algorithm 
from \cite{Burkard+:98} is $O(n^3 + 2^d n^2)$;
in practice, our implementation of this algorithm 
exhibited quadratic growth in running time on most instances.
Both the theoretical and the practical running times were relatively high,
which raised some justifiable doubts 
about the overall superiority of the method.

As it was expected, the introduction 
of the new efficient minimum-weight double-tree algorithm
described in \secref{s-alg}
significantly improved the running time in our computational experiments. 
However, this improvement alone was not sufficient for the algorithm 
to compete against the best existing tour-constructing heuristics. 
Therefore, we introduced two additional heuristic improvements,
one aimed at increasing the algorithm's speed,
the other at improving its approximation quality.

The first heuristic, aimed at speeding up the algorithm,
is suggested by the well-known bounded neighbour lists
\cite[p.~408]{Johnson_McGeoch:02}.
Given a tree, we define the \emph{tree distance} 
between a pair of nodes $a$, $b$,
as the number of edges on the unique path from $a$ to $b$ in the tree.
Given a parameter $k$, the \emph{depth-$k$ list} of node $u$ 
includes all nodes in the subtree $T(u)$ 
with the tree distance from $u$ not exceeding $k$.
The suggested heuristic improvement is to limit the search 
across a subtree rooted at $u$ in \eqref{eq-v}--\eqref{eq-a}
to a depth-$k$ list of $u$ for a suitably chosen value of $k$.
Our experiments suggest 
that this approach improves the running time dramatically,
without a significant negative effect on the approximation quality.

The second heuristic, 
aimed at improving the algorithm's approximation quality, 
works by expanding the space of the tours searched,
in the hope of finding a better solution in the larger space.
Let $T$ be a (not necessarily minimum) spanning tree,
and let $\Lambda(T)$ be the set of all tours conforming to $T$,
i.e.\ the exponential set of all tours 
considered by the double-tree algorithm.
Our goal is to construct a new tree $T_1$,
such that its node degrees are still bounded by a constant, 
but $\Lambda(T) \subsetneq \Lambda(T_1)$.
We refer to the new set of tours as an \emph{enlarged tour neighbourhood}.

Consider a node $u$ in $T$, 
and suppose $u$ has at least one child $v$ which is not a leaf.
We construct a new tree $T_1$ from $T$ 
by applying the \emph{degree-increasing operation},
which makes node $v$ a leaf, and redefines 
all children of $v$ to be children of $u$.
It is easy to check that any tour conforming to $T$ also conforms to $T_1$.
In particular, the nodes of $T(v)$,
which are consecutive in any conforming tour of $T$,
are still allowed to be consecutive in any conforming tour of $T_1$.
Therefore, $\Lambda(T) \subseteq \Lambda(T_1)$.
On the other hand, sequence $w,u,v$, where $w$ is a child of $v$,
is allowed by $T_1$ but not by $T$.
Therefore, $\Lambda(T) \subsetneq \Lambda(T_1)$.

Note that the degree-increasing operation cannot be performed partially:
it would be wrong to reassign only some, instead of all, 
children of node $v$ to a new parent.
To illustrate this statement, suppose that $v$
has two children $w_1$ and $w_2$, which are both leaves. 
Let $w_2$ be redefined as a new child of $u$. 
The sequence $v, w_2, w_1$ is allowed by $T$ but not by $T_1$,
since it violates the requirement for $v$ and $w_2$ to be consecutive.
Therefore, $\Lambda(T) \not\subseteq \Lambda(T_1)$.

We apply the degree-increasing heuristic as follows.
Let $D$ be a global parameter,
not necessarily related to the maximum node degree in the original tree. 
The degree-increasing operation is performed 
only if the resulting new degree of vertex $u$ would not exceed $D$.
Given a tree, the degree increasing operation 
is applied repeatedly to construct a new tree, 
obtaining an enlarged tour neighbourhood. 
In our experiments, we used breadth-first application 
of the degree increasing operation as follows:
\begin{quote}
Root the minimum spanning tree at a node of degree $1$;\\
Let $r'$ denote the unique child of the root;\\
Insert all children of $r'$ into queue $Q$;\\
\textbf{while} queue $Q$ is not empty \textbf{do}\\
\hspace*{1em}extract node $v$ from  $Q$;\\
\hspace*{1em}insert all children of $v$ into $Q$;\\
\hspace*{1em}\textbf{if} $\deg(\mathit{parent}(v)) + \deg(v) \leq D$ 
             \textbf{then}\\
\hspace*{2em}redefine all children of $v$ to be children 
             of $\mathit{parent}(v)$
\end{quote}

To incorporate the described heuristics,
the minimum-weight double-tree algorithm from \secref{s-alg} 
was modified to take two parameters: 
the search depth $k$, and the degree limit $D$.
We refer to the double-tree algorithm with fixed parameters $k$ and $D$ 
as a \emph{double-tree heuristic DT$_{D,k}$}. 
We use DT without subscripts to denote 
the original minimum-weight double-tree algorithm,
equivalent to DT$_{1,\infty}$.

\section{Computational experiments}
\label{s-compexp}

We compared experimentally the efficiency of the original algorithm DT
with the efficiency of double-tree heuristics DT$_{D,k}$
for two different search depths $k=16,32$, 
and for four different values for the degree limit 
$D=1$ (no degree increasing operation applied), $3$, $4$, $5$.
The case $D=2$ is essentially equivalent to $D=1$, 
and therefore not considered.

The DIMACS Implementation Challenge \cite{Johnson_McGeoch:02} provided 
an excellent opportunity for testing and evaluating new approaches to the TSP\@.
Website [DIMACS], created to support the Challenge,
contains a wide range of test instances and experimental data.
In our computational experiments, 
we used uniform random Euclidean instances with 1000 points (10 instances), 
3162 points (five instances), 10000 points (three instances), 
31623 and 100000 points (two instances of each size), 
316228, 1000000, and 3168278 points (one instance of each size). 

\begin{table}[t]
\centering

\subfloat[Average excess over the Held--Karp bound (\%)]{%
\begin{tabular}{|l|rrrrrrrr|}
\hline

Size        & 1000 & 3162 & 10K  & 31K  & 100K & 316K & 1M   & 3M   \\ \hline
DT          & 7.36 & 7.82 & 8.01 & 8.19 & 8.39 & 8.40 & 8.41 & --   \\ \hline
DT$_{1,16}$ & 8.64 & 9.24 & 9.10 & 9.43 & 9.74 & 9.66 & 9.72 & 9.66 \\ \hline
DT$_{3,16}$ & 6.64 & 6.97 & 7.04 & 7.37 & 7.51 & 7.53 & 7.55 & 7.50 \\ \hline
DT$_{3,32}$ & 6.52 & 6.84 & 6.92 & 7.21 & 7.31 & 7.36 & 7.37 & 7.31 \\ \hline
DT$_{4,16}$ & 6.00 & 6.27 & 6.39 & 6.69 & 6.82 & 6.87 & 6.85 & --   \\ \hline
DT$_{4,32}$ & 5.93 & 6.22 & 6.33 & 6.60 & 6.74 & 6.78 & 6.77 & --   \\ \hline
DT$_{5,16}$ & 5.67 & 5.91 & 5.97 & 6.27 & 6.43 & 6.51 & 6.47 & --   \\ \hline
DT$_{5,32}$ & 5.62 & 5.89 & 5.93 & 6.23 & 6.38 & 6.46 & 6.43 & --\\ \hline
\end{tabular}}

\subfloat[Average normalised running time (s)]{%
\begin{tabular}{|l|rrrrrrrr|}
\hline
Size        & 1000 & 3162 & 10K   & 31K  & 100K & 316K & 1M   & 3M   \\ \hline
DT          & 0.18 & 1.56 & 15.85 & 294.38 &  3533& 51147& 156659& -- \\ \hline
DT$_{1,16}$ & 0.04 & 0.14 &  0.47 &   1.57 &  5.60 & 20.82 & 101.09 & 388.52 \\ \hline
DT$_{3,16}$ & 0.10 & 0.33 &  1.12 &   3.55 & 11.90 &  40.91 &  138.41 & 491.58 \\ \hline
DT$_{3,32}$ & 0.18 & 0.69 &  2.45 &   7.56 &  25.46 &  82.99 & 269.73 & 935.55 \\ \hline
DT$_{4,16}$ & 0.23 & 0.84 &  2.78 &   8.81 &  29.02 &  94.36 & 307.31 &  -- \\ \hline
DT$_{4,32}$ & 0.45 & 2.00 &  6.93 &  22.11 &  74.70 & 236.33 & 744.50 & -- \\ \hline
DT$_{5,16}$ & 0.62 & 2.30 &  7.79 &  24.48 &  81.35 & 253.59 & 807.74 & -- \\ \hline
DT$_{5,32}$ & 1.11 & 5.74 & 20.73 &  65.96 & 224.34 & 695.03 & 2168.95 & -- \\ \hline
\end{tabular}}

\caption{\label{tab:heuristics}
Results for DT and DT$_{D,k}$ on uniform Euclidean distances}
\end{table}

For each heuristic, 
we consider both its approximation quality and running time.
We say that one heuristic \emph{dominates} another,
if it is superior in both these respects.
Following the approach of the DIMACS Challenge,
approximation quality is measured 
in terms of the approximate solution's excess over the Held--Karp bound
(the solution to the standard linear programming relaxation of the TSP),
and the running time in terms of the ``normalised computation time'' 
(see \cite{Johnson_McGeoch:02}, [DIMACS] for details).
The experimental results, presented in Table~\ref{tab:heuristics},
clearly indicate that nearly all considered heuristics%
\footnote{Heuristic DT$_{1,32}$ is omitted from Table~\ref{tab:heuristics},
since it does not give any noticeably better results compared to DT$_{1,16}$.}
(excluding DT$_{1,16}$) dominate plain DT.
Moreover, all these heuristics (again excluding DT$_{1,16}$) 
dominate DT on each individual instance used in the experiment.

For further comparison of the double-tree heuristics 
with existing tour-constructing heuristics, 
we chose DT$_{1,16}$ and DT$_{5,16}$.

\begin{table}[p]
\centering

\subfloat[Average excess over the Held--Karp bound (\%)]{%
\begin{tabular}{|l|rrrrrrrr|}
\hline
Size & 1000 & 3162 & 10K  & 31K  & 100K & 316K & 1M   & 3M   \\ \hline
\hline
RA$^+$ & 13.96& 15.25& 15.04& 15.49& 15.43& 15.42& 15.48& 15.47 \\ \hline
Chr-S  & 14.48& 14.61& 14.81& 14.67& 14.70& 14.49& 14.59& 14.51 \\ \hline
FI     & 12.54& 12.47& 13.35& 13.44& 13.39& 13.43& 13.47& 13.49 \\ \hline
Sav    & 11.38& 11.78& 11.82& 12.09& 12.14& 12.14& 12.14& 12.10 \\ \hline
ACh    & 11.13& 11.00& 11.05& 11.39& 11.24& 11.19& 11.18& 11.11 \\ \hline
Chr-G  &  9.80&  9.79&  9.81&  9.95&  9.85&  9.80&  9.79&  9.75 \\ \hline
Chr-HK &  7.55&  7.33&  7.30&  6.74&  6.86&  6.90&  6.79&    -- \\ \hline 
\hline
MTS1 &  6.09& 8.09& 6.23&  6.33&  6.22&  6.20&  -- & -- \\ \hline
MTS3 &  5.26& 5.80& 5.55&  5.69&  5.60&  5.60&  -- & -- \\ \hline
\hline
DT$_{1,16}$& 8.64& 9.24& 9.10& 9.43&  9.74&  9.66&  9.72&  9.66 \\ \hline
DT$_{5,16}$& 5.67& 5.91& 5.97& 6.27&  6.43&  6.51&  6.47&    -- \\ \hline
\end{tabular}}

\subfloat[Average normalised running time (s)]{%
\begin{tabular}{|l|rrrrrrrr|}
\hline
Size & 1000 & 3162 & 10K   & 31K  & 100K & 316K & 1M   & 3M   \\ \hline
\hline
RA$^+$ & 0.06& 0.23&  0.71&  1.9&   5.7&  13&   60&  222 \\ \hline
Chr-S  & 0.06& 0.26&  1.00&  4.8&  21.3&  99&  469& 3636 \\ \hline
FI     & 0.19& 0.76&  2.62&  9.3&  27.7&  65&  316& 1301 \\ \hline
Sav    & 0.02& 0.08&  0.26&  0.8&   3.1&  21&  100&  386 \\ \hline
ACh    & 0.03& 0.12&  0.44&  1.3&   3.8&  28&  134&  477 \\ \hline
Chr-G  & 0.06& 0.27&  1.04&  5.1&  21.3& 121&  423& 3326 \\ \hline
Chr-HK & 1.00& 3.96& 14.73& 51.4& 247.2& 971& 3060&   -- \\ \hline 
\hline
MTS1 & 0.37& 2.56& 17.21& 213.4& 1248& 11834&  -- & -- \\ \hline
MTS3 & 0.46& 3.55& 24.65& 989.1& 2063& 21716&  -- & -- \\ \hline
\hline
DT$_{1,16}$& 0.04& 0.14& 0.47&  1.57&  5.60&  20.82& 101& 389 \\ \hline
DT$_{5,16}$& 0.62& 2.30& 7.78& 24.48& 81.35& 254   & 808&  -- \\ \hline
\end{tabular}}

\caption{\label{tab:mainResults}%
Comparison between established heuristics and DT-heuristics
on uniform Euclidean instances}
\end{table}

\begin{figure}[p]
\centering
\includegraphics[bb=152 483 443 668]{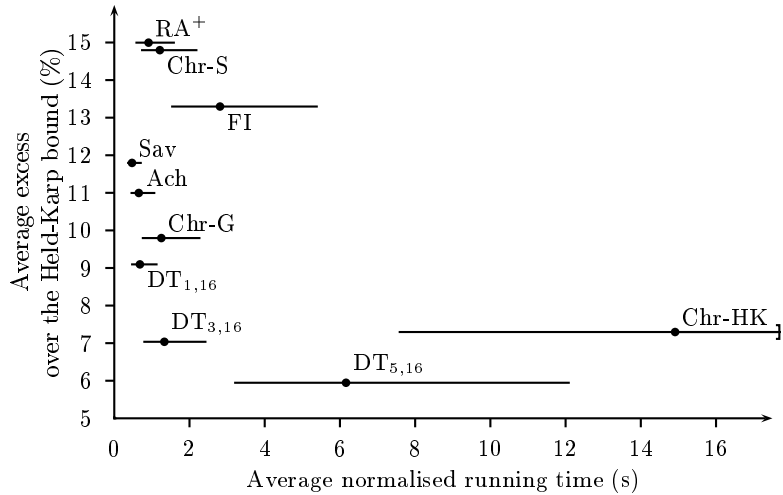}
\caption{\label{f-compare}%
Comparison between established heuristics and DT-heuristics
on uniform Euclidean instances with 10000 points}
\end{figure}

The main part of our computational experiments 
consisted in comparing the double-tree heuristics
against the most powerful existing tour-constructing heuristics.
As a base for comparison, 
we chose the heuristics analysed in \cite{Johnson_McGeoch:02},
as well as two recent matching-based heuristics from \cite{Kahng_Reda:04}.
The experiments were performed on a Sun Systems Enterprise Server E450,
under SunOS 5.8, using the gcc 3.4.2 compiler.

\tabref{tab:mainResults} shows the results of these experiments.
Abbreviations in the table 
follow \cite{Johnson_McGeoch:02,Kahng_Reda:04}:
\begin{itemize}
\item RA$^+$: Bentley's random augmented addition heuristic; 
\item Chr-S: the Christofides heuristic with standard shortcut,
  implemented by Johnson and McGeoch (JM); 
\item FI: Bentley's farthest insertion heuristic; 
\item Sav: saving heuristic, implemented by JM; 
\item ACh: approximate Christofides heuristic, implemented by JM;
\item Chr-G: the Christofides heuristic with greedy shortcut,
  implemented by JM;
\item Chr-HK: the Christofides heuristic on Held--Karp trees 
  instead of MST, implemented by Rohe;
\item MTS1, MTS3: ``match twice and stitch'' heuristics,
  implemented by Kahng and Reda.
\end{itemize}
As seen from the table,
the average approximation quality of DT$_{1,16}$
turns out to be higher than all classical heuristics 
considered in \cite{Johnson_McGeoch:02}, except Chr-HK\@.
Moreover, heuristic DT$_{1,16}$
dominates heuristics RA$^+$, Chr-S, FI, Chr-G\@.
Heuristic DT$_{5,16}$ dominates Chr-HK\@.
Heuristic DT$_{5,16}$ also compares very favourably 
with MTS heuristics, providing similar approximation quality
at a small fraction of the running time.
The above results show clearly
that double-tree heuristics deserve a prominent place 
among the best tour-constructing heuristics for Euclidean TSP.

The impressive success of double-tree heuristics must, however,
be approached with some caution.
Although the normalised time is an excellent tool
for comparing results reported in different computational experiments, 
it is only an approximate estimate of the exact running time. 
According to \cite[page 377]{Johnson_McGeoch:02},
``[this] estimate is still typically 
within a factor of two of the correct time''. 
Therefore, as an alternative way of representing 
the results of computational experiments, 
we suggest a graph of the type shown in \figref{f-compare},
which compares the heuristics' 
average approximation quality and running time
on random uniform instances with 10000 points.
A normalised time $t$ is represented by the interval $[t/2,2t]$. 
The relative position of heuristics in the comparison
and the dominance relationships can be seen clearly from the graph.
Results for other instance sizes and types are generally similar.

\begin{table}[tb]
\centering

\subfloat[Average excess over the Held--Karp bound (\%)]{%
\begin{tabular}{|l|rrrrrr|}
\hline
Size & 1000 & 3162 & 10K  & 31K  & 100K & 316K \\ \hline
\hline
RA$^+$ & 12.84& 13.88& 16.08& 15.59& 16.22& 16.33 \\ \hline
Chr-S  & 12.03& 12.79& 13.08& 13.47& 13.50& 13.45 \\ \hline
FI     &  9.90& 11.85& 12.82& 13.37& 13.96& 13.92 \\ \hline
Sav    & 13.51& 15.97& 17.21& 17.93& 18.20& 18.50 \\ \hline
ACh    & 10.21& 11.01& 11.47& 11.78& 12.00& 11.81 \\ \hline
Chr-G  &  8.08&  9.01&  9.21&  9.47&  9.55&  9.55 \\ \hline
Chr-HK &  7.27&  7.78&  8.37&  8.42&  8.46&  8.56 \\ \hline 
\hline
MTS1   &  8.90& 9.96& 11.97&  11.61&  9.45&  --  \\ \hline
MTS3   &  8.52&  9.5& 10.11&  9.72&  9.46&  --  \\ \hline
\hline
DT$_{4,16}$&  6.37& 8.24& 8.79&  9.40&  9.38&  9.39  \\ \hline
DT$_{5,16}$&  5.72& 7.17& 7.92&  8.32&  8.46&  8.42  \\ \hline
\end{tabular}}

\subfloat[Average normalised running time (s)]{%
\begin{tabular}{|l|rrrrrr|}
\hline
Size & 1000 & 3162 & 10K   & 31K  & 100K & 316K \\ \hline
\hline
RA$^+$ & 0.1& 0.2&  0.7&  1.9&   5.5&  12.7 \\ \hline
Chr-S  & 0.2& 0.8&  3.2& 11.0&  37.8& 152.8 \\ \hline
FI     & 0.2& 0.8&  2.9&  9.9&  30.2&  70.6 \\ \hline
Sav    & 0.0& 0.1&  0.3&  0.9&   3.4&  22.8 \\ \hline
ACh    & 0.0& 0.2&  0.8&  2.1&   6.4&  54.2 \\ \hline
Chr-G  & 0.2& 0.8&  3.2& 11.0&  37.8& 152.2 \\ \hline
Chr-HK & 0.9& 3.3& 11.6& 40.9& 197.0& 715.1 \\ \hline 
\hline
MTS1 &  0.78& 4.19& 45.09&  276&  1798&  --  \\ \hline
MTS3 &  0.84& 4.76& 49.04&  337&  2213&  --  \\ \hline
\hline
DT$_{4,16}$&  0.2& 0.87& 3.16&  9.55&  34.43&  120.3  \\ \hline
DT$_{5,16}$&  1.12& 4.85& 16.08&  53.35&  174&  569  \\ \hline
\end{tabular}}

\caption{\label{tab:mainResultsClustered}
Comparison between established heuristics and DT-heuristics
on clustered Euclidean instances}
\end{table}

\begin{table}[tb]
\centering

\subfloat[Average excess over the Held--Karp bound (\%)]{%
\begin{tabular}{|l|rrrrr|}
\hline
Size & 1000 & 3162 & 10K  & 31K  & 100K \\ \hline
\hline
RA$^+$ & 17.46& 16.28& 17.78& 19.88& 17.39 \\ \hline
Chr-S  & 13.36& 14.17& 13.41& 16.50& 15.46 \\ \hline
FI     & 15.59& 14.28& 13.20& 17.78& 15.32 \\ \hline
Sav    & 11.96& 12.14& 10.85& 10.87& 19.96 \\ \hline
ACh    &  9.64& 10.50& 10.22& 11.83& 11.52 \\ \hline
Chr-G  &  8.72&  9.41&  8.86&  9.62&  9.50 \\ \hline
Chr-HK &  7.38&  7.12&  7.50&  6.90&  7.42 \\ \hline
\hline
MTS1 &  7.0& 6.9& 5.1&  4.7&  4.1 \\ \hline
MTS3 &  6.2& 5.1& 4.0&  2.9&  2.7 \\ \hline
\hline
DT$_{1,16}$&  6.36& 5.99& 8.09&  9.99&  10.02 \\ \hline
DT$_{5,16}$&  6.13& 5.58& 7.65&  8.98&   9.30 \\ \hline
\end{tabular}}

\subfloat[Average normalised running time (s)]{%
\begin{tabular}{|l|rrrrr|}
\hline
Size & 1000 & 3162 & 10K   & 31K  & 100K \\ \hline
\hline
RA$^+$ & 0.1& 0.2& 0.8&  2.2&   5.6 \\ \hline
Chr-S  & 0.1& 0.2& 1.8&  3.9&  31.8 \\ \hline
FI     & 0.2& 0.8& 3.1&  9.8&  26.4 \\ \hline
Sav    & 0.0& 0.1& 0.3&  0.6&   1.4 \\ \hline
ACh    & 0.0& 0.1& 0.5&  1.5&   3.9 \\ \hline
Chr-G  & 0.1& 0.2& 1.8&  3.8&  29.5 \\ \hline
Chr-HK & 0.7& 2.2& 9.7& 50.1& 177.9 \\ \hline 
\hline
MTS1 &  -- & 1.5& 34.4&  107.3& 620.0  \\ \hline
MTS3 &  -- & 2.1& 42.4&  135.4&  1045.3 \\ \hline
\hline
DT$_{1,16}$&  0.3& 0.9& 4.1&  18.4&  49.3 \\ \hline
DT$_{5,16}$&  0.6& 2.1& 11.0&  57.1&  115.1 \\ \hline
\end{tabular}}

\caption{\label{tab:mainResultsTSPLIB}
Comparison between established heuristics and DT-heuristics
on geometric instances from TSPLIB:
pr1002, pcb1173, rl1304, nrw1379 (size 1000),
pr2392, pcb3038, fnl14461 (size 3162),
pla7397, brd14051 (size 10K),
pla33810 (size 31K),
pla859000 (size 100K).}
\end{table}

Additional experimental results for clustered Euclidean instances 
are shown in \tabref{tab:mainResultsClustered}
(with DT$_{1,16}$ replaced by DT$_{4,16}$ 
to illustrate more clearly the overall advantage of DT-heuristics),
and for TSPLIB instances in \tabref{tab:mainResultsTSPLIB}.

While we have done our best 
to compare the existing and the proposed heuristics fairly,
we recognise that our experiments are not, 
strictly speaking, a ``blind test'':
we had the results of \cite{Johnson_McGeoch:02}
in advance of implementing our method, and in particular 
of selecting the top DT-heuristics for comparison.
However, we never consciously adapted our choices 
to the previous knowledge of \cite{Johnson_McGeoch:02},
and we believe that any subconscious effect of this previous knowledge
on our experimental setup is negligible.

\section{Conclusions and open problems}
\label{s-concl}

In this paper, we have presented an improved algorithm
for finding the minimum-weight double-tree shortcutting approximation 
for Metric TSP.
We challenged ourselves to make the algorithm 
as efficient as possible.
The improvement in time complexity from $O(n^3 + 2^d n^2)$ 
to $O(4^d n^2)$ (which implies $O(n^2)$ for the Euclidean TSP)
placed the minimum-weight double-tree shortcutting method as a peer 
in the set of the most powerful tour-constructing heuristics. 
It is known that most such heuristics 
have theoretical time complexity $O(n^2)$,
and in practice often exhibit near-linear running time.
The minimum-weight double-tree method now also fits this pattern.

While we have not been using 
the language of parameterised complexity \cite{Downey_Fellows:98},
we (and the previous work \cite{Burkard+:98}) 
have in fact demonstrated that the problem of finding 
the minimum-weight double-tree tour for Metric TSP
is fixed-parameter tractable 
(where the maximum degree of the MST is the relevant parameter).
It would be interesting to see if this connection 
with parameterised complexity theory can be extended further,
e.g.\ by using any of the established techniques 
for designing fixed-parameter tractable algorithms.

Our results should be regarded only as a first step 
in exploring new opportunities. 
Particularly, the minimum spanning tree 
is not the only possible choice of the initial tree. 
Instead, one can choose from a variety of trees,
e.g.\ Held and Karp (1-)trees, approximations to Steiner trees, 
spanning trees of Delaunay graphs, etc.
This variety of choices merits a further detailed exploration.

It is well-known that when the initial tree is a path,
the resulting double-tree tour neighborhood 
is the set of all \emph{pyramidal tours} \cite{Burkard+:98}. 
In this case, a dozen of conditions on the distance matrix 
are known (see e.g. \cite{Burkard+:98_SIAM}),
which guarantee that the tour neighbourhood contains 
the absolute minimum-weight tour.
It may be possible to generalise this approach
by identifying new special types of trees 
and conditions on the distance matrices, 
which would guarantee that the minimum-weight double-tree algorithm
finds an absolute minimum-weight tour.
For more results on polynomial solvability of TSP
with special conditions imposed on the distance matrix, 
see \cite{Burkard+:98_SIAM,Deineko+:06}.

The minimum-weight shortcutting problem for the Christofides graph
remains NP-hard even in the planar Euclidean metric.
However, our algorithm turns out to be applicable also to this problem
on certain classes of instances.
It can be shown that if the Christofides graph is a \emph{cactus}
(i.e.\ all its cycles are pairwise edge-disjoint),
then the set of all its shortcuttings 
is a subset of the set of all double-tree shortcuttings.
Therefore, our algorithm, as well as the algorithm of \cite{Burkard+:98},
can be used to find efficiently the minimum-weight shortcutting
when the Christofides graph is a cactus.
In particular, such a shortcutting can be found in polynomial time
in the planar Euclidean metric.

Our efforts invested into theoretical improvements of the algorithm, 
supported by a couple of additional heuristic improvements, 
have borne the fruit: 
computational experiments with the minimum-weight double-tree algorithm 
show that it becomes one of the best known tour constructing heuristics.
It appears that the double-tree method is also well suited
for local search improvements based of transformations of trees 
and searching the corresponding tour neighborhoods.
One can easily imagine various tree transformation techniques
that could make our method even more powerful.

\section{Acknowledgements}

The authors thank an anonymous referee of a previous version of this paper,
whose detailed comments helped to improve it significantly.
The MST subroutine in our code 
is courtesy of the Concorde project [Concorde].


\bibliographystyle{acmtrans}
\bibliography{auto,tsp,books}
\nocite{DIMACS,Concorde}

\end{document}